%% file: hl-sat.tex
\documentclass[envcountsame,envcountsect,orivec]{llncs}

\usepackage{\jobname}

\title{%
  The Complexity of Satisfiability for Fragments of Hybrid Logic --- Part I%
  \thanks{Supported in part by the grants DFG VO 630/6-1, DFG SCHW 678/4-1, BC-ARC 1323, DAAD-ARC D/08/08881.}
}

\author{Arne Meier${}^1$, Martin Mundhenk${}^2$, Thomas Schneider${}^3$,\newline
  Michael Thomas${}^1$, Volker Weber${}^4$, Felix Weiss${}^2$}

\institute{%
  ${}^1$Theoretical Computer Science, University of Hannover, Germany           \\
  \protect\url{{meier,thomas}@thi.uni-hannover.de}                    \smallskip\\
  ${}^2$Institut f\"ur Informatik, Universit\"at Jena, Germany                  \\
  \protect\url{{martin.mundhenk,felix.weiss}@uni-jena.de}             \smallskip\\
  ${}^3$Computer Science, University of Manchester, UK                          \\
  \protect\url{schneider@cs.man.ac.uk}                                \smallskip\\
  ${}^4$Fakult\"at f\"ur Informatik, Technische Universit\"at Dortmund, Germany
}

\begin{document}
    
  \maketitle

  \begin{abstract}
    The satisfiability problem of hybrid logics with the downarrow binder is known to be undecidable. This initiated a research program on decidable and tractable fragments.
    
    In this paper, we investigate the effect of restricting the propositional part of the language on decidability
    and on the complexity of the satisfiability problem over arbitrary, transitive, total frames, and frames based on equivalence relations. We also consider different sets of modal and hybrid operators.
    We trace the border of decidability and give the precise complexity of most fragments, in particular for all fragments including negation. For the monotone fragments, we are able to distinguish the easy from the hard cases, depending on the allowed set of operators. 
    \smallskip
    
    \keywordname\ hybrid logic, satisfiability, decidability, complexity, Post's lattice
  \end{abstract}

  \section{Introduction}
  \label{sect:intro}

    Hybrid logics are well-behaved extensions of modal logic. However, their expressive power often has adverse effects
    on their computational properties: for instance, the satisfiability problem for basic modal logic
    extended with the \dna binder is undecidable \cite{blse95,gor96,ArBM99b}, as opposed to \PSPACE-complete
    for basic modal logic \cite{lad77} and modal logic extended with nominals and the satisfaction 
    operator \at\ \cite{ArBM99b}.

    In order to regain decidability, many restrictions of the hybrid binder language have been considered.
    On the syntax side, it has been shown in \cite{cafr05} that restricting the interactions between
    \dna and universal operators (such as $\wedge$, $\Box$) makes satisfiability decidable again.
    On the semantics side, the satisfiability problem for the \dna language has been investigated over
    different frame classes. It becomes decidable over frames with bounded width \cite{cafr05},
    over transitive and complete frames \cite{mssw05}, and over frames with an equivalence relation
    \cite{musc07}. In the latter case, decidability is not lost if \at\ or the global modality is added 
    to the language \cite{musc07}, which is not the case over transitive frames \cite{mssw05}.
    Furthermore, over linear frames and transitive trees, where \dna on its own is useless, extensions of the \dna
    language have been shown to be decidable, albeit nonelementarily, in \cite{frrisc03,mssw05}. But elementarily decidable fragments over these frame classes have been obtained by bounding the number of state variables \cite{SchwentickW07,Weber07,BozzelliL08}.
    An overview of complexity results for hybrid logics can be found in \cite{sch07b}.

    Our aim is to obtain a more fine-grained distinction between decidable and undecidable
    hybrid logics by restricting the set of Boolean operators allowed in formulae.
    This is interesting in its own right because it will outline sources of ``bad'' behaviour
    (i.e., undecidability) more precisely. Furthermore, it is interesting 
    in view of the relation between modal and description logic (DL).
    Concept satisfiability, the DL-counterpart of modal satisfiability, plays an important role
    because other useful decision problems for DLs are reducible to it.
    For a number of DLs without full Boolean expressivity,
    notably the $\mathcal{EL}$ and DL-Lite families, this problem is tractable \cite{Baa03,BaBL05,CdGL+05},
    and other relevant
    decision problems have lower complexity than for the standard DL $\mathcal{ALC}$,
    the counterpart of the modal logic \textsf{K}.
    In the case of these restricted DLs, there are also fine-grained analyses of
    additional features that increase complexity and those which do not \cite{BaBL05,ACKZ07}.
    Our study can be seen as a general framework which accommodates restrictions of different types---on
    Boolean operators systematically, and also on modal operators and frame classes.
    As one possible application of the obtained results, we will gain insights into the complexity
    of extensions of modal and description logics with hybrid operators---among them the above mentioned restricted DLs.

    For the sake of generality,
    we will systematically replace the usual $\wedge$, $\neg$ with arbitrary, not necessarily complete, sets of Boolean
    operators. All such possible sets are captured
    in Post's lattice \cite{pos41,bocrrevo03}, which consists of all clones, i.\,e., all closed sets of Boolean
    functions. Each clone corresponds to a set of Boolean operators closed under
    nesting, and vice versa. The lattice allows for transferring upper and lower
    complexity bounds between clones. It will thus be possible to prove
    finitely many results that will be valid for an infinite number of sets of
    operators---and hence for infinitely many satisfiability problems.
    This technique has been used for analysing the complexity of
    satisfiability for propositional logic \cite{lew79} and modal logic \cite{bhss05b},
    satisfiability and model checking for linear temporal logic \cite{bsssv07,bamuscscscvo07},
    and satisfiability of constraint satisfaction problems \cite{sch78,sch07a}.

    Using Post's lattice, we will investigate the complexity of the satisfiability problem
    for hybrid logics containing the modal operators $\Diamond,\Box$ and
    the following hybrid features: nominals, the satisfaction operator \at\ and
    the hybrid binder \dna. We will consider subsets of these operators,
    as well as the above described systematic restrictions to the Boolean operators allowed.
    We will carry out this analysis over four different frame classes: all frames, transitive frames,
    total frames (where every state has at least one successor), and frames with equivalence relations (ER frames).
    The work presented here is part of ongoing work that also includes acyclic frame classes such as transitive trees and linear structures.

    While our analysis is complete with respect to the sets of Boolean operators covered,
    it is far from complete for sets of modal and hybrid operators, as well as for frame classes.
    This is because the latter ``dimensions'' of expressivity are much more difficult to systematise.
    Therefore, we are currently restricting ourselves to the most prominent sets of modal/hybrid operators
    and frame classes. It should also be noted that a fourth dimension is possible, namely allowing
    for multiple accessibility relations in models, i.e., multiple modalities of each kind.
    We have omitted this consideration from the present paper mostly for the sake of a clearer presentation.
    However, we believe that many of the upper bounds can be straightforwardly extended to the multi-modal
    case---and will therefore indeed be helpful to gain insights into the behaviour of more expressive
    description logics.

    This paper contains the most complete subset of our results obtained so far (see Figure \ref{fig:hlsat}),
    namely the following. We will show that, over each of the four above frame classes,
    satisfiability is as hard as in the full Boolean case whenever the negation of the implication
    or self-dual Boolean operators are allowed. (A Boolean function is self-dual if
    negating all of its arguments will always negate its value.) This means that, in these cases,
    satisfiability remains undecidable over arbitrary frames, total frames and, if the \at-operator is present, over transitive frames;
    and \NEXPTIME-complete over transitive frames without \at\ and over ER frames. These results can be found in Section \ref{subsubsec:all}.

\begin{figure}[tb]
        \centering
        \includegraphics[width=\textwidth]{full_lattice.1} 
        \caption{\label{fig:hlsat}%
                Post's Lattice. 
                The complexity of $\SaT[\Fclass{F}](O,B)$ for frame classes 
                $\Fclass{F}\in\{\mathsf{all}, \mathsf{trans}, \mathsf{total}, \mathsf{ER}\}$ 
                and sets $O$ of Boolean operators with 
                $\{\Diamond,\dna\}\subseteq O\subseteq\{\Diamond,\Box,\dna,\at\}$.%
        }
\end{figure}

		In Section \ref{subsubsec:negation}, we also completely classify the complexity of fragments including only negation and the Boolean constants. We obtain completeness for \LOGSPACE if the \at-operator is included and for $\ACp{0}{2}$ otherwise.
		
		For all monotone fragments including the Boolean constant $\false$, we obtain a duality between easy cases, which are all included in \NC1, and hard cases, for which we obtain lower bounds ranging from \LOGSPACE to \PSPACE (Section \ref{subsec:monotone}). Satisfiability for fragments not including $\false$, but possibly all $\true$-reproducing functions, 
turns out to be trivial as shown in Section \ref{subsec:r1}.
		              
 For the fragments that are based on the binary \textsf{xor} operator, the complexity is open. This case has
    turned out to be difficult to handle in \cite{bhss05b,bsssv07,bamuscscscvo07}. 
A list of still open questions can be found in Section \ref{sec:conclusions}.

  \section{Preliminaries}
  \label{sect:prelim}

  \ourparagraph{Boolean Functions and Clones}
    We can identify an $n$-ary propositional operator (connector) $c$ with the $n$-ary \emph{Boolean function} $f_c \colon \{\false,\true\}^n\to \{\false,\true\}$ defined by $f_c(a_1,\ldots,a_n)=\true$ if and only if $c(x_1,\ldots,x_n)$ becomes true when assigning $a_i$ to $x_i$ for all $1 \leq i \leq n$. The Boolean values \emph{false} and \emph{true} correspond to constants, i.\,e., nullary functions, and will be denoted by $\false$ and $\true$.

    A set of Boolean functions is called a \emph{clone} if it contains all projections and is closed under arbitrary composition \cite[Chapter~1]{pip97b}. The set of all Boolean clones forms a lattice, which has been completely classified by Post \cite{pos41}.
    For a set $B$ of Boolean functions, we denote by $[B]$ the smallest clone containing $B$ and call $B$ a \emph{base} for $[B]$.
    Whenever we use $B$ for a set, we assume that $B$ is finite.

    In order to introduce the clones relevant to this paper, we define the following notions
    for $n$-ary Boolean functions $f$:
    \begin{itemize} \itemsep 0pt
      \item $f$ is \emph{$t$-reproducing} if $f(t, \ldots , t) = t$, $t \in \{\false,\true\}$.
      \item $f$ is \emph{monotone} if $a_1 \leq b_1, \ldots , a_n \leq b_n$ implies $f(a_1, \ldots , a_n) \leq f(b_1, \ldots , b_n)$.
      \item $f$ is \emph{$t$-separating} if there exists an $i \in \{1, \ldots , n\}$ such that $f(a_1, \ldots , a_n) = t$ implies $a_i = t$, $t \in \{\false,\true\}$.
      \item $f$ is \emph{self-dual} if $f \equiv \dual{f}$, where $\dual{f}(x_1, \ldots , x_n) = \neg f(\neg x_1, \ldots , \neg x_n)$.
    \end{itemize}
    The clones relevant to this paper are listed in Table \ref{tab:clones}. 
    The definition of all Boolean clones can be found, e.\,g., in \cite{bocrrevo03}.
    Notice that $[B\cup\{\true\}]=\CloneBF$ if and only if  $[B] \supseteq \CloneS_1$ or $[B] \supseteq \CloneD$.

    \begin{table}[tbh]
      \centering
      \small
      \begin{tabular}{c|l|l}
        Name & Definition & Base \\
        \hline
        $\CloneBF$ & All Boolean functions & $\{\land, \neg\}$ \\
        $\CloneR_1$ & $\true$-reproducing functions & $\{\lor, \limplies \}$ \\
        $\CloneM$ & monotone functions & $\{\lor, \land, \false, \true\}$ \\
        $\CloneS_1$ & $\true$-separating functions & $\{x \land \overline{y}\}$ \\
        $\CloneS_{11}$ & $\CloneS_1 \cap \CloneM$ & $ \{x \land (y \lor z), \false \} $ \\
        $\CloneD$ & self-dual functions & $\{ (x \land \overline{y}) \lor (x \land \overline{z}) \lor (\overline{y} \land \overline{z}) \}$ \\
        $\CloneV$ & constant or $n$-ary OR functions & $\{ \lor, \false,\true \}$ \\
        $\CloneE$ & constant or $n$-ary AND functions  & $\{ \land, \false, \true \}$ \\
        $\CloneE_0$ &            & $\{ \land, 0 \}$ \\
        $\CloneN$ & functions depending on at most one variable  & $\{ \neg,\false,\true\}$ \\
        $\CloneN_2$ &  & $\{ \neg\}$ \\
        $\CloneI$ & constant or identity functions & $\{\id, \false,\true\}$ \\
        $\CloneI_0$ &                 & $\{\id, \false\}$ \\
        $\CloneI_1$ &                 & $\{\id, \true\}$ \\
        $\CloneI_2$ &                 & $\{\id\}$ \\
      \end{tabular}
      \smallskip        
      \caption{\label{tab:clones}Boolean clones relevant to this paper, with definitions and bases.}
    \end{table}

  \ourparagraph{Hybrid Logic}
    In the following, we will introduce the notions and definitions of hybrid logic. The terminology is largely taken from \cite{arblma00}.

    Let $\PROP$ be a countable set of \emph{atomic propositions}, $\NOM$ be a countable set of \emph{nominals}, $\SVAR$ be a countable set of \emph{variables} and $\ATOM = \PROP \cup \NOM \cup \SVAR$. We will stick with the common practice to denote atomic propositions by $p,q,\ldots$, nominals by $i,j,\ldots$, and variables by $x,y,\ldots$.
    We define the language of \emph{hybrid (modal) logic} $\HL$ as the set of well-formed formulae of the form
    \[
      \varphi ::= a \mid  c(\varphi,\ldots,\varphi) \mid \Diamond \varphi \mid \Box \varphi \mid \dna x.\varphi  \mid \at_t \varphi
    \]
    where $a \in \ATOM$, $c$ is a Boolean operator, $x \in \SVAR$ and $t \in \NOM \cup \SVAR$.
    Note that the usual cases $\top$ and $\bot$ are covered by the Boolean constants \true and \false.

    Formulae of $\HL$ are interpreted on \emph{(hybrid) Kripke structures} $K=(W,R,\eta)$, consisting of a set of \emph{states} $W$, a \emph{transition relation} $R\colon W\times W$, and a \emph{labeling function} $\eta\colon \PROP\cup\NOM \to\wp(W)$ that maps $\PROP$ and $\NOM$ to subsets of $W$ with $|\eta(i)| = 1$ for all $i \in \NOM$. In order to evaluate $\dna$-formulae, an assignment $g\colon \SVAR \to W$ is necessary. Given an assignment $g$, a state variable $x$ and a state $w$, an \emph{$x$-variant $g^x_w$ of $g$} is defined by $g^x_w(x)=w$ and $g^x_w(x')=g(x')$ for all $x \neq x'$. For any $a \in \ATOM$, let $[\eta,g](a)=\{g(a)\}$ if $a \in \SVAR$ and $[\eta,g](a)=\eta(a)$, otherwise.
    The satisfaction relation of hybrid formulae is defined by

    \begin{tabbing}
    \hspace*{\parindent}\= $K,g,w \models c(\varphi_1,\ldots,\varphi_n)$ \= iff \= $K,g^x_w,w$ \=  $p \in \eta(w)$, $p \in \PROP \cup \NOM$ \kill
      \> $K,g,w \models a$ \> iff \> $w \in [\eta,g](a)$, $a \in \ATOM$, \\
      \> $K,g,w \models c(\varphi_1,\ldots,\varphi_n)$ \> iff \> $f_c(t_1,\ldots,t_n)=\true$, where $t_i$ is the truth value of \\
      \>                   \>     \> $K,g,w$ \> $\models \varphi_i$, $1 \leq i \leq n$, \\
      \> $K,g,w \models \Diamond \varphi$ \> iff \> $K,g,w'$ \> $\models \varphi$ for some $w' \in W$ with $wRw'$,  \\
      \> $K,g,w \models \Box \varphi$ \> iff \> $K,g,w'$ \> $\models \varphi$ for all $w' \in W$ with $wRw'$, \\
      \> $K,g,w \models \at_t \varphi$ \> iff \> $K,g,w$ \> $\models \varphi$ for $w \in W$ such that $w \in \eta(t)$, \\
      \> $K,g,w \models \dna x. \varphi$ \> iff \> $K,g^x_w,w$ \> $\models \varphi$. 
    \end{tabbing}

    A hybrid formula $\varphi$ is said to be \emph{satisfiable} if there exists a Kripke structure $K=(W,R,\eta)$, a $w \in W$ and an assignment $g\colon \SVAR \to W$ such that $K,g,w \models \varphi$.

    The \emph{at} operator $\at_t$ shifts evaluation to the state named by $t\in \NOM \cup \SVAR$.
    The \emph{downarrow binder} $\dna x.$ binds the state variable $x$ to the current state. 
    The symbols $\at_x$, $\dna x.$  
    are called \emph{hybrid operators} whereas the symbols $\Diamond$ and $\Box$ are called \emph{modal operators}. 

    For considering fragments of hybrid logics, we define subsets of the language $\HL$ as follows.
    Let $B$ be a finite set of Boolean functions and $O$ a set of hybrid and modal operators. We define $\HL(O,B)$ to denote the set of well-formed hybrid formulae using the operators in $O$ and the Boolean connectives in $B$ only.

  \ourparagraph{Properties of Frames}
    A \emph{frame} $F$ is a pair $(W,R)$, where $W$ is a set of states and $R\subseteq W\times W$ a transition relation.
    We will refer to a frame as being \emph{transitive}, \emph{total} or \emph{ER} whenever its transition relation $R$ is transitive
    ($uRv \land vRw \limplies uRw$), total ($\forall u \exists v (uRv)$), or an equivalence relation, i.\,e.,
    reflexive ($uRu$), transitive and symmetric ($uRv \limplies vRu$).
    In this paper we will consider
      the class $\mathsf{all}$ of all frames,
      the class $\mathsf{trans}$ of all transitive frames,
      the class $\mathsf{total}$ of all total frames, and
      the class $\mathsf{ER}$ of all ER frames.

  \ourparagraph{The Satisfiability Problem}
    Let $K=(W,R,\eta)$ be a Kripke structure.
    Say that \emph{$K$ is based on a frame $F$} iff $F$ is the frame underlying $K$, i.\,e., $F=(W,R)$.
    We define the \emph{satisfiability problems} for the fragments of $\HL$ over frame classes defined above as follows.

    \problemdef%
      {$\SaT[\Fclass{F}](O,B)$}
      {an $\HL(O,B)$-formula $\varphi$}
      {is there a Kripke structure $K=(W,R,\eta)$ based on a frame from $\Fclass{F}$, an assignment $g\colon \SVAR \to W$ and a $w\in W$ such that $K,g,w\models\varphi$\,?}
    In case $\Fclass{F}=\mathsf{all}$, we will omit the prefix and simply write $\SaT(O,B)$.

  \ourparagraph{Complexity Theory}
    We assume familiarity with the standard notions of complexity theory as, e.\,g., defined in \cite{pap94}.
    In particular, we will make use of the classes $\L$, $\NL$, $\P$, $\coNP$, $\PSPACE$, $\NEXP$, and $\coRE$.
    
    We will now introduce the notions of circuit complexity required for this paper,
    for more information on circuit complexity the reader is referred to \cite{vol99}.
    The class $\NC{1}$ is defined as the set of languages recognizable by a logtime-uniform Boolean
    circuits of logarithmic depth and polynomial size over $\{\land,\lor,\neg\}$,
    where the fan-in of $\land$ and $\lor$ gates is fixed to $2$.

    The class $\AC0$ is defined as the set of languages recognizable by a logtime-uniform Boolean
    circuits of constant depth and polynomial size over $\{\land,\lor,\neg\}$,
    where the fan-in of gates of the first two types is not bounded. If, in addition, modulo-2 gates are allowed,
    then the corresponding class is $\ACp02$. Both $\AC0$ and $\ACp02$ are strictly contained in $\NC1$.
    Altogether, the following inclusions are known:
    $
        \AC0 \subseteq \ACp02 \subset \NC1 \subseteq \L \subseteq \NL \subseteq \P \subseteq
        \coNP \subseteq \PSPACE \subset \NEXP \subset \coRE.
    $

A language $A$ is \emph{constant-depth reducible} to $D$, $A\leqcd D$, if there is a logtime-uniform $\AC{0}$-circuit family with oracle gates for $D$ that decides membership in $A$. Unless otherwise stated, all reductions in this paper are \leqcd-reductions.

  \ourparagraph{Known results}
    The following theorem summarizes results for hybrid binder languages with Boolean operators $\wedge,\vee,\neg$
    that are known from the literature.

    \begin{theorem}[\cite{ArBM99b,mssw05,musc07}]
      \label{theo:known_for_and_or_neg}
      ~\par
      \begin{Enum}
        \item
          $\SaT(\{\Diamond,\dna\}, \{\wedge,\vee,\neg\})$ and $\SaT(\{\Diamond,\Box,\dna,\at\}, \{\wedge,\vee,\neg\})$ are \coRE-complete.
        \item
          $\SaT[trans](\{\Diamond,\dna\}, \{\wedge,\vee,\neg\})$ is \NEXP-complete.
       \item
          $\SaT[trans](\{\Diamond,\dna,\at\}, \{\wedge,\vee,\neg\})$ is \coRE-complete.
        \item
          $\SaT[ER](\{\Diamond,\dna\}, \{\wedge,\vee,\neg\})$ is \NEXP-complete.
        \item
          $\SaT[ER](\{\Diamond,\dna,\at\}, \{\wedge,\vee,\neg\})$ is \NEXP-complete.
      \end{Enum}
    \end{theorem}

%

  \section{Results}
  \label{sect:results}
	
	In this section, we present our results ordered by clones. Section \ref{subsec:r1} considers clones containing only \true-reproducing functions. Clones containing the Boolean constant \false but not negation are considered in Section \ref{subsec:monotone}. Finally, in Section \ref{subsec:negation}, we study satisfiability problems based on clones with negation.

This arrangement is motivated by the observation that the availability of the Boolean constant \false and/or negation has a very strong impact on our results. Although we obtain different complexities for the clones including $\false$ but not negation (namely, $\CloneI$, $\CloneV$, $\CloneE$, and $\CloneM$), the results for these clones follow a certain pattern. But if we add negation, this picture changes completely.

Please note that opposed to the importance of the presence of \false, which makes the difference between trivial and nontrivial problems, hybrid languages can always express the constant \true as $\dna x.x$ or $\at_x x$. Therefore, we only have to consider clones including \true.

	\input{R1}

	\input{I-V-E-M}

	\input{N-BF}

  \section{Conclusions}\label{sec:conclusions}

    We have almost completely classified the complexity of hybrid binder logics over four frame classes
    with respect to all possible combinations of Boolean operators, see Figure \ref{fig:hlsat}.
    The main open question is for tight upper bounds for the monotone fragments including the $\Box$-operator over the classes of all and of transitive frames.
    
    Another open questions concerns the hybrid languages with $\Box$ but without $\at$ over the class of transitive frames. The complexity for the respective satisfiability problems based on $\CloneV$, $\CloneE$, and $\CloneM$ is open; in the case of $\CloneV$ even for the class of all frames. For $\CloneI$, containment in $\AC0$ follows from an analysis of the proof of Theorem \ref{theo:all,trans,total,ER_N_Dia,Box,dna}.
    Finally, we could not obtain any bounds on the complexity for problems based on $\CloneL$, besides \LOGSPACE-hardness inherited from Theorem \ref{thm:hl-sat-Dia,Box,dna,at-I}.

    We are currently investigating the same problems over frame classes important for representing
    modal properties, such as transitive trees, linear frames and the natural numbers.
    Here, satisfiability for \dna, \at, and arbitrary Boolean operators is already decidable,
    but with a nonelementary lower bound; hence, a complexity analysis is worthwile as well.
    Because each such frame is acyclic, the fact that certain formulae are always satisfied in
    the singleton reflexive frame is not helpful any longer. This makes obtaining upper bounds more difficult.
    On the other hand, we can also express the constant \false by $\dna x.\Diamond x$, which reduces
    the sets of Boolean operators to consider. We plan to publish these results in ``Part II''.

  \bibliographystyle{plain}
  \bibliography{thi-hannover}

\end{document}

%% file: R1.tex
\subsection{Why we cannot say anything without saying ``false''}\label{subsec:r1}

We start our investigation at the clone $\CloneI_2$, which contains only the identity function.\footnote{Please remember that we can always express the Boolean constant $\true$ by $\dna x.x$. Hence, there is no difference between the satisfiability problems for $\CloneI_2$ and $\CloneI_1$.} Obviously, every hybrid $\CloneI_2$-formula is satisfied by the model consisting of a singleton reflexive state to which all propositions, nominals, and state variables are labeled.

But this observation takes us much further, as we can add conjunction, disjunction, and implication for example, and still satisfy every formula by the same model. In fact, we can add every $\true$-reproducing function, i.\,e., every function that produces \true if all parameters are \true, obtaining the following result.

  \begin{theorem}\label{thm:Rone}
    $\SaT[\Fclass{F}](\{\Diamond,\Box,\dna,\at\}, B)$ for $[B] \subseteq \CloneR_1$ and all considered frame classes $\Fclass{F}\in\{\mathsf{all}, \mathsf{trans}, \mathsf{total}, \mathsf{ER}\}$ is trivial.
  \end{theorem}

  \begin{proof}
     All Boolean functions in the clone $\CloneR_1$ are $\true$-reproducing,
         hence every propositional $\CloneR_1$-formula $\varphi$ is satisfiable.
     It is easily seen that every modal $\CloneR_1$-formula is satisfiable
     by the singleton reflexive Kripke structure $K=(\{w\},\{(w,w)\},\eta)$ with $\eta(p)=\{w\}$ for all $p \in \PROP \cup \NOM$, which is included in all frame classes we consider here.
       As $K$ includes only one state and every $\HL(\{\Diamond,\Box\},\CloneR_1)$-formula is satisfied in $w$,
     bindings and jumps do not change satisfiability. Therefore, all hybrid $\CloneR_1$-formulae are satisfiable.
  \end{proof}

It is interesting to note which Boolean operations are not contained in $\CloneR_1$. The most basic ones are the Boolean constant \false and negation, as every clone in Post's lattice that is not below $\CloneR_1$ contains one of these.

As hybrid languages can always express the Boolean constant \true, the presence of negation implies the availability of \false.
Therefore, there are two kinds of clones remaining: those containing \false but not negation, and those containing negation.
In the following subsection, we will consider the first kind, i.\,e., the monotone clones below $\CloneM$. Clones with negation will be considered in Section \ref{subsec:negation}.

%% file: I-V-E-M.tex
\subsection{Everything but negation -- The monotone clones}\label{subsec:monotone}

In this section, we consider the clones below $\CloneM$ that contain the Boolean constant \false; satisfiability for the clones without \false is trivial by Theorem \ref{thm:Rone}. Roughly speaking, we consider the clones $\CloneI$, $\CloneV$, $\CloneE$, and $\CloneM$. We start with $\CloneI$ and then jump to $\CloneM$. Clones containing either disjunction or conjunction are considered last, as some results will easily follow from the preceding cases.

\subsubsection{The clone \CloneI}

The clone $\CloneI$ is of particular interest, as it allows us to study the effect of having the Boolean constant \false at our disposal, yielding the following two observations.
First, the Boolean constant \false distinguishes trivial from nontrivial satisfiability problems. While all satisfiability problems for clones without \false are trivial (Theorem~\ref{thm:Rone}), all problems for clones with \false are not. The precise complexity of the latter problems will vary from almost trivial cases (Theorem~\ref{thm:I-trivial-cases}) to \L-completeness (Theorem~\ref{thm:hl-sat-Dia,Box,dna,at-I}), depending on the modal and hybrid operators allowed. Higher complexities and even undecidability occur if we add further Boolean functions as discussed in the following sections.

Second, Theorems \ref{thm:I-trivial-cases} and \ref{thm:hl-sat-Dia,Box,dna,at-I} demonstrate a duality between easy and hard cases, which we will see in all results for clones below $\CloneM$. For the full set of modal and hybrid operators, satisfiability problems over the class of all frames and the class of transitive frames will be considerably harder than those over total frames and equivalence relations. Furthermore, if we drop the $\Box$-operator when considering arbitrary or transitive frames, complexity will drop to where it is for total frames and equivalence relations.

Intuitively speaking, we might say that the complexity gap we observe in the results for monotone clones is due to the ability to express that a state has no successor by $\Box\false$. On the one hand, if we cannot express this property because of the absence of $\Box$ or if there are no such states because we only consider frames with a total accessibility relation, satisfiability for the clone $\CloneI$ is \emph{almost trivial}, i.\,e., we only need to look at one symbol of a formula to determine its satisfiability.
  \begin{theorem}\label{thm:I-trivial-cases}
    The following satisfiability problems are almost trivial.\footnote{More precisely, they are in $\Delta^\mathcal{R}_0$, a class strictly below \DLOGTIME \cite{revo97}.}
    \begin{enumerate}
      \item \label{trivial:Ia}
         $\SaT[total](\{\Diamond,\Box,\dna,\at\},B)$ and $\SaT[ER](\{\Diamond,\Box,\dna,\at\},B)$ for $[B] \subseteq \CloneI$.
      \item  \label{trivial:Ib}
         $\SaT(\{\Diamond,\dna,\at\},B)$ and $\SaT[trans](\{\Diamond,\dna,\at\},B)$ for $[B] \subseteq \CloneI$.
    \end{enumerate}
  \end{theorem}

  \begin{proof}
      Every formula in $\HL(\{\Diamond,\dna,\at\},\CloneI)$ 
        consists of a sequence of operators followed eventually by one final symbol from $\ATOM\cup\{0,1\}$.
      For all considered frame classes, these formulae are satisfiable if and only if this final symbol is not $0$.
      For the frame classes \textsf{total} and \textsf{ER}, this also holds if we add $\Box$ to the operators allowed.
  \end{proof}

On the other hand, the proof of the following theorem shows how to use $\Box\false$ to obtain \LOGSPACE-hardness, without using any further Boolean connectives. A matching upper bound will be presented in Theorem \ref{thm:hl-sat-total-all-n}.
  \begin{theorem}\label{thm:hl-sat-Dia,Box,dna,at-I}
    $\SaT(\{\Diamond,\Box,\dna,\at\},B)$ and $\SaT[trans](\{\Diamond,\Box,\dna,\at\},B)$ for $[B]\supseteq \CloneI_0$
     are $\L$-hard.
  \end{theorem}

  \begin{proof}
   We give a reduction from the problem \emph{Order between Vertices} (\ORD) which is known to be $\L$-complete \cite{etessami:1997}.

    \problemdef%
        {$\ORD$}
        {a finite set of vertices $V$, a successor-relation $S$ on $V$, and two vertices $s,t\in V$}
        {is $s\leq_S t$, where $\leq_S$ denotes the unique total order induced by $S$ on $V$?}

    Notice that $(V,S)$ is a directed line-graph.
    Let $(V,S,s,t)$ be an instance of $\ORD$. 
    We construct a $\HL(\{\Diamond,\Box,\dna,\at\}, \CloneI_0)$-formula $\varphi$ 
    that is satisfiable if and only if $s\leq_S t$. 

    We use $V=\{v_1,\ldots,v_n\}$ as state variables. The formula $\varphi$ consists of three parts. 
    The first part binds all variables except $s$ to one state and the variable $s$ to a successor state.
    The second part of $\varphi$ binds a state variable $v_l$ to the state labeled by $s$ iff $s\leq_S v_l$. 
    Let $\alpha$ denote the concatenation of all $\at_{v_{k}}\dna v_{l}$ with $(v_k,v_l)\in S$ and $v_k\not=s$,
     and $\alpha^n$ denotes the $n$-fold concatenation of $\alpha$.
    Essentially, $\alpha^n$ uses the assignment collect all $v_i$ with $s\leq_S v_i$ in the state labeled $s$.
    \begin{claim}
      $K,g,u\models \alpha^n\, \at_x\psi$ iff
      $K,g',u\models \at_x\psi$ for $g'$ with $g'(v_i)=g(s)$ for all $v_i\geq_S s$ and $g'(v_i)=g(v_i)$ for all $v_i\not\geq_S s$.
    \end{claim}

    \noindent
    It is not hard to prove the Claim.
    The last part of $\varphi$ guarantees that $s$ was initially bound to another state than the remaining variables,
    and checks whether $s$ and $t$ are bound to the same state after this procedure.
    \[
      \varphi = \dna v_1.\dna v_2. \cdots \dna v_n. \Diamond \dna s. ~ \alpha^n ~ \at_t \Box \false
    \]
%
    To prove the correctness of our reduction, we show that $\varphi$ is satisfiable if and only if $s\leq_S t$. 
    If $s\leq_S t$, then for $K=(\{u,w\}, \{(u,w)\},\eta)$ with arbitrary $\eta$ and $g$ it holds that $K,g,u\models \varphi$.
    For $s\not\leq_S t$, consider any $K$ with state $u$.
    We show that $K,g,u\not\models\varphi$.
    Let $g_1$ be the assignment obtained from $g$ after the bindings of the first part
    $\dna v_1.\dna v_2. \cdots \dna v_n. \Diamond \dna s$ of $\varphi$,
    and let $g'_1$ be the assignment obtained from $g_1$ after the first part and the second part $\alpha^n$.
    By the Claim it follows that $g'_1(t)=g_1(t)=\{u\}$.
    If $u$ has no successor, then the first part $\dna v_1.\dna v_2. \cdots \dna v_n. \Diamond \dna s$ of $\varphi$
    is not satisfied.
    If $u$ has a successor, then $K,g'_1,u\not\models \Box \false$,
    and it follows that $K,g'_1,u\not\models \at_t \Box \false$ and therefore $K,g,u\not\models\varphi$.
  \end{proof}

\subsubsection{The clone \CloneM}

Let us now consider the clone $\CloneM$ of all monotone functions. Here, more precisely for all clones between $\CloneS_{11}$ and $\CloneM$, we obtain the same duality as in the previous section, only at a higher level of complexity.
For the ``hard cases'', i.\,e., those satisfiability problems where we consider non-total frame classes and all modal and hybrid operators, we obtain \PSPACE-hardness. For the class of all frames, this follows immediately from the corresponding result for modal logic.
\begin{theorem}[\cite{bhss05b,hem01}]
       $\SaT(\{\Diamond,\Box\}, B)$ is \PSPACE-hard
       for $[B]\supseteq \CloneS_{11}$.
\end{theorem}

Unfortunately, the proof of this result does not generalize to transitive frames. 

  \begin{lemma} \label{lem:hl-sat-Dia,Box,dna,at-m}
    $\SaT[trans](\{\Diamond,\Box,\dna,\at\},\{\land,\lor,\false\})$ is $\PSPACE$-hard.
  \end{lemma}

  \begin{proof}
  $\QBFSAT$ is a standard $\PSPACE$-complete set.
  Its instances are quantified Boolean formulae in conjunctive normal form,
  e.\,g.\ 
  {

    \centering

    $\varphi_0 = \exists x_1 \forall x_2 \exists x_3 \forall x_4 (x_1 \vee \neg x_2) \wedge (\neg x_1 \vee x_2 \vee x_3 \vee \neg x_4)$.

  }
  A quantified Boolean formula is in $\QBFSAT$ if and only if it evaluates to true.
  We give a reduction that reduces $\QBFSAT$ to $\SaT[trans](\{\Diamond,\Box,\dna,\at\},\{\land,\lor,\false\})$.
  An instance $\varphi=\exists x_1 \forall x_2\cdots Q x_n \psi$
  is transformed to 
   \[f(\varphi)=(\at_s\, \Diamond\Diamond 1) \wedge (\at_s\, \Diamond \Box \false) \wedge  
                \at_s\Diamond \dna x_1 . \at_s \Box \dna x_2. \cdots \at_s Q' \dna x_n. r(\psi),\]
  where $Q'=\Diamond$ ($Q'=\Box$) if $Q=\exists$ ($Q=\forall$),
  and $r(\psi)$ is obtained from $\psi$ by replacing all appearances of \emph{positive} literals $x_i$ with $\at_{x_i} \Diamond \true$,
  respectively replacing all appearances of \emph{negative} literals $\neg x_i$ with $\at_{x_i} \Box \false$.
  Constant symbols remain unchanged.
  As an example, the $\QBFSAT$ instance $\varphi_0$ from above is transformed to
  {

   \centering

   ~~~~~~~~$(\at_s\, \Diamond\Diamond \true) \wedge (\at_s\, \Diamond\Box \false) \wedge  \at_s \Diamond \dna x_1. \at_s \Box \dna x_2. 
               \at_s \Diamond \dna x_3. \at_s \Box \dna x_4.$\hfill\mbox{}\\
   \hfill
        $(\at_{x_1} \Diamond \true \vee \at_{x_2} \Box \false) \wedge 
              (\at_{x_1} \Box \false \vee \at_{x_2} \Diamond \true \vee \at_{x_3} \Diamond \true \vee \at_{x_4} \Box \false)$.
  }\\
  It is clear that  $f$ is a $\leqcd$-reduction.

  The following claim implies $\QBFSAT \leqcd \SaT[trans](\{\Diamond,\Box,\dna,\at\},\{\land,\lor,\false\})$.
  \begin{claim}
  $\varphi$ evaluates to true if and only if 
  $K_2,g,s \models f(\varphi)$ for $K_2=(W,R,\eta)$
      with the transitive frame $(W,R)=(\{s,t\}, \{(s,s),(s,t)\})$ and $\eta(x_i)=W$.
  \end{claim}
  
  \noindent
  The proof of the Claim is straightforward.
  \end{proof}

\PSPACE-hardness also follows for all clones containing $\CloneS_{11}$.
    \begin{theorem}
      \label{thm:hl-sat-Dia,Box,dna,at-m}
       $\SaT[trans](\{\Diamond,\Box,\dna,\at\},B)$ is \PSPACE-hard
       for $[B]\supseteq \CloneS_{11}$.
    \end{theorem}

    \begin{proof}
      It suffices to show $\SaT[trans](O, \{\land,\lor,\false\}) \leqPm \SaT[trans](O, B)$. The result then follows from Lemma \ref{lem:hl-sat-Dia,Box,dna,at-m}.

      Take $\varphi \in \HL(O, \{\land,\lor,\false\})$.
      Since $[B \cup \{\ONE\}] = \CloneM$,
      we can rewrite $\varphi$ as an $\HL(O, B \cup \{\ONE\})$-formula $\varphi'$,
      leaving modal and hybrid operators untouched.
      Due to \cite{lew79,sch07a}, this  can be computed in polynomial time.
      Now we can easily transform $\varphi'$ into an $\HL(O, B)$-formula $\varphi''$,
      replacing all occurrences of \ONE with $\at_z z$.
      Clearly, $\varphi$ and $\varphi''$ are equisatisfiable over transitive (and even over arbitrary) frames.
    \end{proof}

  The proof of Theorem \ref{thm:hl-sat-Dia,Box,dna,at-m} crucially depends on the existence of states without successor, 
  and the ability to express this property: 
  the truth values $\bot$ (resp.\ $\top$) are encoded as states having no (resp.\ at least one) successor.
  If there are no such states (Theorem \ref{thm:hl-sat-total,er-all-m}) or if we cannot express this property (Corollary \ref{cor:hl-sat-all,trans-Dia,dna,at-m}), complexity drops to $\NC1$.

   \begin{theorem} \label{thm:hl-sat-total,er-all-m}
    $\SaT[total](O, B)$ and $\SaT[ER](O, B)$
    are $\NC1$-complete under 
    $\leqcd$-reductions for $O \subseteq \{\Diamond,\Box,\dna,\at\}$ and $\CloneS_{11} \subseteq [B] \subseteq \CloneM$.
  \end{theorem}

  \begin{proof}
    Any propositional $\CloneM$-formula $\varphi$ is satisfiable if and only if it is satisfied by the assignment
    that sets all atoms to true.
    We generalize this result to hybrid logic 
    for frame classes $\mathsf{total}$ and $\mathsf{ER}$.   

    Let $\varphi \in \HL(\{\Diamond,\Box,\dna,\at\}, \CloneM)$ and let $K_1:=(\{w_1\},\{(w_1,w_1)\},\eta_1)$ be the reflexive singleton Kripke structure 
    with $\eta_1(p)=\{w_1\}$ for all $p \in \PROP \cup \NOM$. 
    The following claim is easy to verify.
\begin{claim}
    $K_1,g_1,w_1\models\psi$ is equivalent to
    $K_1,g_1,w_1\models \lambda \psi$ 
    for any operator $\lambda$.
\end{claim}

    \noindent
    We show that $\varphi \in \SaT[total](\{\Diamond,\Box,\dna,\at\}, \CloneM)$ if and only if $K_1,g_1,w_1 \models \varphi$.
    We proceed by induction on the structure of $\varphi$.
    If $\varphi\in \PROP \cup \NOM\cup \SVAR$, then  $\varphi\in\SaT[total](\{\Diamond,\Box,\dna,\at\}, \CloneM)$ and $K_1,g_1,w_1 \models \varphi$.

    For the inductive step, assume that the claim holds for all subformulae of $\varphi$.
    \begin{itemize}
      \item $\varphi = \Box \psi$ ($\varphi = \Diamond \psi$).
         For total frames, $\Box \psi$ ($\Diamond \psi$) is satisfiable if and only if $\psi$ is satisfiable.
         By induction hypothesis, this is equivalent to $K_1,g_1,w_1\models \psi$,
         and by the Claim this is equivalent to $K_1,g_1,w_1\models \Box \psi$ ($K_1,g_1,w_1\models\Diamond \psi$).
      \item $\varphi = \dna x. \psi$.
       Then $\varphi\in \SaT[total](\{\Diamond,\Box,\dna,\at\}, \CloneM)$ if and only if 
        there exists $K,g,w$ such that $K,g^x_w,w \models \psi$.
        By induction hypothesis and by the Claim this is equivalent to $K_1,g_1,w_1\models \dna x. \psi$.
      \item The remaining cases for hybrid operators follow similarly.

      \item $\varphi = c(\psi_1,\ldots, \psi_n)$ with $c \in \CloneM$.
       Assume that $\varphi$ is satisfied by some Kripke structure $K$ under assignment $g$ in state $w$. 
       By induction hypothesis we obtain:
        if $K,g,w \models \psi_i$ then $K_1,g_1,w_1 \models \psi_i$.
       Since $c \in \CloneM$, it follows that $K,g,w~\models~c(\psi_1,\ldots,\psi_n)$ if and only if $K_1,g_1,w_1 \models c(\psi_1,\ldots, \psi_n)$.
          \end{itemize}
    This shows that deciding $\varphi \in \SaT[total](\{\Diamond,\Box,\dna,\at\}, \CloneM)$ is equivalent to deciding whether $K_1,g_1,w_1\models\varphi$. In order to decide the latter,
    all hybrid and modal operators of $\varphi$ can be ignored, as $K_1$ is a singleton model.
    Thus deciding $K_1,g_1,w_1\models\varphi$\linebreak is equivalent to the evaluation problem for propositional $\CloneM$-formula, 
    which is\linebreak[4] $\NC1$-complete under $\leqcd$-reductions \cite{sch07a}.

    The same arguments apply for $\SaT[ER](\{\Diamond,\Box,\dna,\at\}, \CloneM)$.
  \end{proof}

The proof of Theorem~\ref{thm:hl-sat-total,er-all-m}, shows that deciding $\varphi \in \SaT[total](\{\Diamond,\Box,\dna,\at\}, \CloneM)$
is equivalent to deciding whether $K_1,g_1,w_1\models\varphi$, for the singleton reflexive model $K_1$ mapping all atomic propositions into state $w_1$. In order to decide the latter, all hybrid and modal operators of $\varphi$ can be ignored, as $K_1$ is a singleton model.
There, only the treatment of the $\Box$-operator depends on the transition relation being total or an equivalence relation. 
If this operator is not allowed, the same argumentation goes through for our other frame classes, too. 

  \begin{corollary} \label{cor:hl-sat-all,trans-Dia,dna,at-m}
    $\SaT(O,\!\!\;B)$ and $\SaT[trans](O,\!\!\;B)$
    are $\NC1\!\!\;$-complete for $O \subseteq \{\Diamond,\!\!\;\dna,\!\!\;\at\}$ and $\CloneS_{11} \subseteq [B] \subseteq \CloneM$.

  \end{corollary}

\subsubsection{The clones \CloneV and \CloneE}

If we consider conjunction or disjunction only separately, the complexity of formula evaluation decreases from $\NC1$-complete to below $\AC0$. As the complexity for the ``easy cases'' for $\CloneM$ was determined by this complexity (Theorem~\ref{thm:hl-sat-total,er-all-m} and Corollary~\ref{cor:hl-sat-all,trans-Dia,dna,at-m}), the following results are not too surprising.

  \begin{theorem} \label{thm:hl-sat-cyclic-all-e+v}
    $\SaT[total](\{\Diamond,\Box,\dna,\at\},B)$ and $\SaT[ER](\{\Diamond,\Box,\dna,\at\},B)$ are in $\AC0$ for $[B] \subseteq \CloneE$ or $[B] \subseteq \CloneV$.
  \end{theorem}

  \begin{proof} 
  As $\CloneV \subseteq \CloneM$ and $\CloneE \subseteq \CloneM$, the hybrid operators $\dna$ and $\at_x$ may be ignored as in the proof of Theorem~\ref{thm:hl-sat-total,er-all-m}. It has been shown by Schnoor~\cite{sch07a}, that evaluation of $\CloneV$-formulae ($\CloneE$-formulae) is in $\co\NLOGTIME$ ($\NLOGTIME$, resp.). Since $\NLOGTIME\cup\co\NLOGTIME\subseteq\AC0$, the theorem applies.
  \end{proof}
We obtain the following result from Theorem~\ref{thm:hl-sat-cyclic-all-e+v}.
	
  \begin{corollary} \label{cor:hl-sat-all,trans-Dia,dna,at-e+v}
    $\SaT(O, B)$ and $\SaT[trans](O, B)$
    are in $\AC0$ for $O \subseteq \{\Diamond,\!\!\;\dna,\!\!\;\at\}$ and $[B] \subseteq \CloneE$ or $[B] \subseteq \CloneV$.
  \end{corollary}

This result is optimal in the sense that including all modal and hybrid operators we immediately get \LOGSPACE lower bounds from Theorem \ref{thm:hl-sat-Dia,Box,dna,at-I}. 

For the case of conjunctions, 
this result can be improved. Considering arbitrary frames, a \coNP lower bound is already known for the modal satisfiability problem.
\begin{theorem}[\cite{bhss05b,dhlnns92}]
    $\SaT(\{\Diamond,\Box\}, B)$ is $\coNP$-hard for $[B] \supseteq \CloneE_0$.
\end{theorem}

As before, the proof of this result does not generalize to transitive frames. Here, we are able to show \NL-hardness. 

  \begin{theorem} \label{thm:hl-sat-trans-all-e}
    $\SaT[trans](\{\Diamond,\Box,\dna,\at\},B)$ is $\NL$-hard for $[B] \supseteq \CloneE_0$.
  \end{theorem}

  \begin{proof}
    We give a reduction from the unreachability problem for acyclic graphs. 
    For acyclic graphs, there is a path from $s$ to $t$ if and only if $s$ appears before $t$ 
    in every topological sorting of the nodes.
    We make use of the negation of this statement.

    Let $G=(V,E,s,t)$ be an instance of unreachability for acyclic digraphs with $V=\{1,\ldots,n\}$.
    We construct a $\HL(\{\Diamond,\Box,\dna,\at\}, \{\wedge,\false\})$-formula $\varphi = \varphi_1 \wedge \varphi_2$ 
        that is satisfiable if and only if $G$ has no path from $s$ to $t$. 

        We will use a state variable $x_i$ for each node $i$ and an additional state variable $r$ that will be used for the ``root state''.
        The first part of $\varphi$ is used to guarantee that every Kripke structure $K$ that satisfies $\varphi$
        consists of at least $n+1$ states with an acyclic transition relation.
        \[
          \varphi_1 = \Diamond^{n}\true \land \Box^{n+1} \false
        \]

        The second part of $\varphi$ gives names to the states, such that the order of the states reflects a topological
        order of the corresponding nodes of the graph, and checks whether $s$ is behind $t$.
        \[
          \varphi_2 = \dna r.\at_r\Diamond\dna x_1. \cdots \at_r\Diamond\dna x_n. \bigwedge_{(i,j)\in E}( \at_{x_i} \Diamond x_j)\land\at_{x_t}\Diamond x_s
        \]
        Note that it does not matter here if more than one variable is assigned to one state as long as the edge-relation is respected.

        It is not too hard to see that $\varphi$ is satisfiable (over a transitive frame) if and only if there exists a topological ordering of $V$ in which $t$ appears before $s$. 
    Consequently, $\SaT[trans](\{\Diamond,\Box,\dna,\at\}, \CloneE_0)$ is $\NL$-hard.
  \end{proof}

We conjecture that all lower bounds provided in this section (except, perhaps, the last one) are optimal. 
Nevertheless, matching upper bounds are missing.

\begin{theorem}
  $\SaT(\{\Diamond,\Box,\dna,\at\},\CloneV)$ is in \LOGSPACE.
\end{theorem}

\begin{proof}
  Let $\varphi \in \HL(\{\Diamond,\Box,\dna,\at\},\CloneV)$, and let $a_1,\dots,a_n$ be all atoms (i.e., occurrences of
  \ZERO, \ONE, atomic propositions, nominals or state variables) in $\varphi$. Given a subformula $\psi$ of $\varphi$
  and $i$ such that $a_i$ occurs in $\psi$, we define $\psi[a_i]$ to be the following formula.
  \begin{itemize}
    \item
      If $\psi = a_i$, then $\psi[a_i] = a_i$.
    \item
      If $\psi = \vartheta_1 \vee \vartheta_2$, then $\psi[a_i] = \vartheta_j[a_i]$, where $j$ is such that $\vartheta_j$ contains $a_i$.
    \item
      If $\psi = O \vartheta$, for $O \in \{\Diamond,\Box,\dna x, \at_x\}$, then $\psi[a_i] = O \vartheta[a_i]$.
  \end{itemize}
  Clearly, each $\psi[a_i]$ is in $\HL(\{\Diamond,\Box,\dna,\at\},\CloneI)$.

  The following claim yields a decision procedure for $\SaT(\{\Diamond,\Box,\dna,\at\},\CloneV)$ that is
  based on a procedure for $\SaT(\{\Diamond,\Box,\dna,\at\},\CloneI)$, which can be found in Theorem \ref{thm:hl-sat-total-all-n}.

  \begin{claim}
    $\varphi$ is satisfiable if and only if there is some $i = 1,\dots,n$ such that $\varphi[a_i]$ is satisfiable.
  \end{claim}

  Since $\SaT(\{\Diamond,\Box,\dna,\at\},\CloneI)$ is in \LOGSPACE, the claim enables us to decide
  $\SaT(\{\Diamond,\Box,\dna,\at\},\CloneV)$ in \LOGSPACE as follows: for each $a_i$, test whether
  $\varphi[a_i]$ is satisfiable. If one of these tests is positive, accept; otherwise reject. The administrational
  effort of traversing through all atoms and determining the respective $\varphi[a_i]$ involves only
  constantly many counters of logarithmic size and determining scopes of operators. Therefore,
  the whole decision procedure can be performed in logarithmic space.

  It remains to prove the claim.
  The ``$\Leftarrow$'' direction is obvious: if $\varphi[a_i]$ is satisfiable in a state $s$ of a
  Kripke structure $K$ under an assignment $g$, then $K,g,s$ also satisfies $\varphi$.
  For the ``$\Rightarrow$'' direction, we recall hybrid tableau techniques as described in \cite{bla00a}.
  If $\varphi$ is satisfiable, then there exists a tableau with $\at_i \varphi$ at its root, for a fresh nominal $i$,
  and an open and complete branch in this tableau. This branch by itself is a tableau for $\varphi[a_i]$,
  and therefore $\varphi[a_i]$ is satisfiable.
\end{proof}

%% file: N-BF.tex
\subsection{Clones including negation}\label{subsec:negation}

Negation immediately limits the number of relevant satisfiability problems in two ways. First, as $\Box\varphi\equiv\neg\Diamond\neg\varphi$, we cannot exclude the $\Box$-operator and keep $\Diamond$ as we did for monotone clones. Therefore, we have to consider only two hybrid languages: with and without $\at$.
Second, as \true and \false are always expressible by $\dna x.x$ and $\neg\dna x.x$, we only need to consider clones with both constants. These are $\CloneN$ (only negation), $\CloneL$ (exclusive or), and $\CloneBF$ (all Boolean functions).

While we will completely classify all satisfiability problems based on $\CloneN$ and $\CloneBF$, we will not provide any specific results for $\CloneL$.

\subsubsection{Negation only}\label{subsubsec:negation}

The results for the satisfiability problems based on $\CloneN$ stick out from our other results, as $\CloneN$ is the only clone (besides those for which satisfiability is trivial) where all complexity results are the same for all frame classes we consider. We show that satisfiability for the hybrid language including $\at$ is \LOGSPACE-complete, while it is \ACp{0}{2}-complete for the language without $\at$.
%
We start with the lower bound for the latter result, which holds even in the absence of modal operators.
    \begin{theorem}
      \label{theo:all,trans,total,ER_N_dna}
      $\SaT[\Fclass{F}](\{\dna\},B)$ is  \ACp{0}{2}-hard for $[B]\supseteq \CloneN_2$
      and \allconsideredF.
    \end{theorem}
    \begin{proof}
    We give a reduction from the \ACp{0}{2}-complete problem \PARITY. 

    Let $a:=a_1\cdots a_n\in \{0,1\}^n$ be an instance for \PARITY. Then the following $\HL(\{\dna\},\CloneN_2)$-formula is satisfiable if and only if $a$ contains an even number of ones.
    \[\varphi_a := \dna x.b_1\dots b_n x,\]
    where $b_i=\neg$ if $a_i=1$ and the empty string otherwise. That is, the number of negations in $\varphi_a$ is exactly the number of ones in $a$. As the satisfiability of $\varphi_a$ does not depend on the transition relation, the theorem follows.
    \end{proof}

A matching upper bound is provided by the following theorem.
    \begin{theorem}
      \label{theo:all,trans,total,ER_N_Dia,Box,dna}
      $\SaT[\Fclass{F}](\{\Diamond,\Box,\dna\},B)$ is \ACp{0}{2}-complete
      for $[B] \subseteq \CloneN$ 
      and \allconsideredF.
    \end{theorem}
    The key to the proof of this theorem is that a given formula can be transformed into negation normal form by an \ACp{0}{2}-circuit. Subsequently determining satisfiability is easy.

  \begin{proof}
      Let $\varphi \in \HL(\{\Diamond,\Box,\dna\},B)$. Before we consider the different frame classes, we observe that we can assume that $\varphi$ is of the form
      \begin{align*}
              (\{\Diamond,\Box\}\cup\{\dna x. \mid x\in\SVAR\})^\star\lambda,
      \end{align*}
      where $\lambda\in\{\lnot x, x, \lnot p, p, \lnot n, n, \true, \false\}$ for a variable $x\in\SVAR$, a proposition $p\in\PROP$, and a nominal $n\in\NOM$. In particular, there is at most one negation and if there is a negation, it is the second to last symbol of $\varphi$. If $\varphi$ is not given in this form, an equivalent formula can be computed by simply drawing negations inward, replacing operators by their duals and eliminating double negations. More precisely, we have to compute for every operator in $\varphi$ whether it is in the scope of an odd or even number of negations. This can be done by an \ACp{0}{2}-circuit.

      Now, assume that $\varphi$ is of the form described above, and let us distinguish two cases.
      First, assume that $\lambda$ is not the negation of a bound state variable.
      In this case, it is rather trivial to determine satisfiability of~$\varphi$: 
      over total or ER frames, $\varphi$ is unsatisfiable if and only if $\lambda=\false$;
      over arbitrary or transitive frames, $\varphi$ is unsatisfiable if and only if $\lambda=\false$ and $\varphi$ contains no $\Box$-operator.

      Second, if $\lambda=\neg x$ for some state variable $x$, satisfiability of $\varphi$ depends only on the temporal operators between the last 
      $\dna$-operator binding $x$ and $\lambda$. If there are no temporal operators in between, then $\varphi$ is clearly unsatisfiable;
      otherwise, we need to distinguish between frame classes.

      Over arbitrary, transitive and total frames, $\varphi$ is satisfied in the model obtained from the natural
      numbers with the usual order $<$ as transition relation if and only if there is a temporal operator in between.
      Over ER frames, $\varphi$ is satisfiable if and only if $\varphi$ is satisfiable in a complete frame,
      i.e., one were the transition relation contains every possible edge ($R=W\times W$). 
      In such a model, satisfiability of $\varphi$ depends only on the last temporal operator. 
      If this operator is a $\Diamond$, then $\varphi$ is satisfiable in model containing at least two states. 
      If the last temporal operator is a $\Box$, then $\varphi$ is unsatisfiable.

      Therefore, the complexity of satisfiability is in all cases dominated by the translation of the given formula into the form described above, and hence in \ACp{0}{2}.
  \end{proof}

We now turn to the hybrid language including the $\at$-operator.

    \begin{theorem}\label{thm:hl-sat-total-diamond-n-hard}
    $\SaT[\Fclass{F}](\{\Diamond,\dna,\at\}, B)$ is $\L$-hard for $[B] \supseteq \CloneN_2$
    and all considered frame classes $\Fclass{F}\in\{\mathsf{all}, \mathsf{trans}, \mathsf{total}, \mathsf{ER}\}$.
    \end{theorem}

    \begin{proof}
    For $\Fclass{F}\in\{\mathsf{all}, \mathsf{trans}\}$ this follows from Theorem~\ref{thm:hl-sat-Dia,Box,dna,at-I}.
    To capture the other frame classes as well, we give a reduction from the problem \emph{Order between Vertices} (\ORD) similar to the one in the proof of Theorem \ref{thm:hl-sat-Dia,Box,dna,at-I}.
    Let $\alpha$ be the string obtained from $(V,S,s,t)$ as in the proof of Theorem~\ref{thm:hl-sat-Dia,Box,dna,at-I}.
    Then
    \[
      \varphi=\dna x_1.\Diamond\dna x_2.\Diamond \dna x_3.\cdots\Diamond\dna x_n. ~ \alpha^n~ \at_{s} \neg t .
    \]
    From $K,w,g\models\varphi$ it follows that $s\not\leq_S t$
    using the claim in the proof of Theorem~\ref{thm:hl-sat-Dia,Box,dna,at-I}.
    If  $s\not\leq_S t$ then every Kripke structure 
    in which $s$ and $t$ can be bound to different states satisfies $\varphi$.
    This proves the correctness of our reduction.
    \end{proof}

We again provide a matching upper bound for all considered frame classes, which yields \LOGSPACE-completeness of the respective satisfiability problems.
    \begin{theorem}\label{thm:hl-sat-total-all-n}
      $\SaT[\Fclass{F}](\{\Diamond,\Box,\dna,\at\}, B)$ is in $\L$ for $[B] \subseteq \CloneN$ and 
      all considered frame classes $\Fclass{F}\in\{\mathsf{all}, \mathsf{trans}, \mathsf{total}, \mathsf{ER}\}$.
    \end{theorem}

  \begin{proof}
      As in the proof of Theorem \ref{theo:all,trans,total,ER_N_dna}, we can assume that the given formula $\varphi\in\HL(\{\Diamond,\Box,\dna,\at\},\CloneN)$ is  of the form
      \begin{align*}
        \varphi = & \big(\{\Diamond,\Box\}\cup\{\at_t\mid t\in\NOM\cup\SVAR\}\cup\{\dna x. \mid x\in\SVAR\}\big)^\star\lambda,
      \end{align*}
      where $\lambda\in\{\lnot x, x, \lnot p, p, \lnot i, i, \true, \false\}$ for a variable $x\in\SVAR$, 
      a proposition $p\in\PROP$, and a nominal $i\in\NOM$. We consider each frame class separately.

      \par\smallskip\noindent
      \textbf{Over total frames, }%
      to decide whether $\varphi \in \SaT[total](\{\Diamond,\Box,\dna,\at\}, B)$, we distinguish the following three cases.
      \newcommand{\modstring}[1]{\mathord{(#1)_{\vartriangle}}}
      \begin{description}
        \item[$\lambda\in\{\false, \true\}$:] As the transition relation is total, we obtain that $\varphi\in\SaT[total]$ iff $\lambda=\true$.
        
        \item[$\lambda\in\{p,\lnot p,x,i\}$:] In this case, $\varphi$ is satisfied by the singleton reflexive model.
        
        \item[$\lambda\in\{\lnot x,\lnot i\}$:] For the last case that $\lambda$ is a negated nominal or a negated state variable. 
        First consider the case that $\varphi$ is free of unbound state variables and nominals, i.e., $\lambda = \neg x$ for some $x \in \SVAR$.
        We can assume w.l.o.g.\ that each state variable is bound at most once. This can be achieved by bound renaming, 
        which is computable in logarithmic space because it only involves computing scopes of binders and counting occurrences
        of \dna and \at\ operators.

        Let $\overline{\ATOM}:= \{\neg a \mid a  \in \ATOM\}$.
        We define a transformation $\modstring\cdot$ as follows:
        For a formula $\psi=P_1 \cdots P_n a$ with $a \in \ATOM \cup \overline{\ATOM}$ and $P_i \in \{\dna x. ,\at_{t},\Diamond,\Box \mid x \in \SVAR, t \in \NOM \cup \SVAR\}$, let $(\psi_j)_{j \in \N}$ be the family of strings obtained as follows:
            $\psi_0 := \psi$, and 
            $\psi_{j+1}$ is derived from $\psi_j$ 
            by deleting the substring $P_\ell \cdots P_k$, $1 \leq \ell < k \leq n$, from $\psi_j$ where 
            $P_\ell = \dna x.$, $P_k =\at_x$ and
            such that the suffix $P_{k+1}\cdots P_n$ is free of $\at$-operators.
        Note that, in the first step, such an $\ell$ exists by assumption.
        Obviously, there is some $\tilde{\jmath} \leq n$ such that $\psi_{\tilde{\jmath}} = \psi_{\tilde{\jmath}+1}$.
        Intuitively, $\psi_{\tilde{\jmath}}$ is obtained by reading $\psi$ from right to left, 
        discarding all symbols between, and including, any $\at_x$-operator and the corresponding $\dna x.$
        We define $\modstring{\psi}$ as $\psi_{\tilde{\jmath}}$ with all remaining $\dna x.$-operators removed.
        Obviously, $\modstring{\psi} \in \{\Diamond, \Box\}^\star \cdot (\ATOM \cup \overline{\ATOM}) $.
        For example, if $\psi:= \Diamond\Box \dna x.\Diamond\Diamond \dna y. \Box \at_x \Diamond \dna z.\Box \at_z y$, then
        $\modstring{\psi}=\Diamond\Box \Diamond y = P_1P_2P_9y$.

        Let $\varphi=P_1 \cdots P_n \neg x$ denote the given formula and let $1 \leq m \leq n$ be such that $P_m = \dna x$. We define $\varphi'$ as $P_1 \cdots P_m \true$.
        The formulae $\modstring{\varphi}$ and $\modstring{\varphi'}$ describe the sequence of modalities relevant 
        to reach the state(s) that must not be labeled with $x$, 
        resp.\ the sequence of modalities relevant to reach the state that is actually labeled by $x$.
        It remains to check that the sequences described by $\modstring{\varphi}$ and $\modstring{\varphi'}$ are ``compatible''.
        Denote by $\kappa$ the maximal index such that $i_{\kappa} = j_{\kappa}$.
        Then $P_1 \cdots P_{i_\kappa}$ describes the common prefix of both formulae that may be ignored w.\,r.\,t.\ the satisfiability of $\varphi$.
        We therefore omit $P_1 \cdots P_{i_\kappa}$ from both sequences
        $\modstring{\varphi'}$ and $\modstring{\varphi}$, i.e., we set $\modstring{\varphi'}=P_{j_\kappa+1} \cdots P_{j_\nu} \true$
        and $\modstring{\varphi}=P_{i_\kappa+1} \cdots P_{i_\mu} \neg x$ 
        with $1 \leq i_{\kappa+1} < \cdots < i_\mu \leq n$ and $1 \leq j_{\kappa+1} < \cdots < j_\nu < m$.
        It now holds that $\varphi$ is satisfiable if and only if $\modstring{\varphi'}$ and $\modstring{\varphi}$ 
        meet one of the following conditions:
        \begin{enumerate}
          \item \label{total-sat-N-item1} 
          $\mu \neq \nu$, 
          
          \item \label{total-sat-N-item2}
          $\mu=\nu$ and there exists a $\xi$, $\kappa < \xi \leq \mu$,
          such that $P_{i_\xi} = P_{j_\xi} = \Diamond$, 
          
          \item \label{total-sat-N-item3}
          $\mu=\nu$ and there exists a $\xi$, $\kappa < \xi \leq \mu$,
          such that $P_{i_\xi} = \Diamond$, $P_{j_\xi} = \Box$, $j_\xi < i_\xi$ (i.\,e., the position of $P_{i_\xi}$ in $\varphi$ 
          is to the right of $P_{j_\xi}$),
          
          \item \label{total-sat-N-item4}
          $\mu=\nu$ and there exists a $\xi$, $\kappa < \xi \leq \mu$,
          such that $P_{i_\xi} = \Box$, $P_{j_\xi} = \Diamond$, $i_\xi < j_\xi$ (i.\,e., the position of $P_{j_\xi}$ in $\varphi$ 
          is the right of $P_{i_\xi}$).
        \end{enumerate}
        
        \par\smallskip\noindent
        \emph{For necessity of (\ref{total-sat-N-item1})--(\ref{total-sat-N-item4}),}
        let $\varphi$ be an $\HL(\{\Diamond,\Box,\dna,\at\},\CloneN)$-formula without nominals and unbound state variables
        which is satisfiable in a total frame but does not satisfy any of \eqref{total-sat-N-item1}--\eqref{total-sat-N-item4}.
        From the converses of conditions \eqref{total-sat-N-item1} and \eqref{total-sat-N-item2}, 
        we derive that $\mu = \nu$ and that for all $\kappa < \xi \leq \mu$,
        $P_{i_\xi}=\Box$ or $P_{j_\xi}=\Box$. 
        Now, observe that in total structures if $P_{i_\xi}=P_{j_\xi}=\Box$, 
        then $\varphi$ and $\varphi$ with $P_{i_\xi}$ and $P_{j_\xi}$ removed are equisatisfiable.

        Hence we may assume that for all $\kappa < \xi \leq \mu$, either $P_{i_\xi}=\Box$ or $P_{j_\xi}=\Box$, but not both.
        Moreover, the converses of \eqref{total-sat-N-item3} and \eqref{total-sat-N-item4} imply $i_\xi > j_\xi$ for the former and $j_\xi > i_\xi$ for the latter case. But if we evaluate $\varphi$ on some total Kripke structure $K$, then every $P_{i_\xi}=\Diamond$ (resp.\ $P_{j_\xi}=\Diamond$) is processed prior to the corresponding modality $P_{j_\xi}=\Box$ (resp.\ $P_{i_\xi}=\Box$). This is a contradiction to the satisfiability of $\varphi$:
        suppose $P_{i_\xi}=\Diamond$, $P_{j_\xi}=\Box$, and  $i_\xi < j_\xi$. Then $P_{i_\xi}=\Diamond$ enforces a successor $v$
        behind which $x$ is bound, but $P_{j_\xi}=\Box$ causes the subformula starting at $P_{j_\xi}$ to be satisfied in $v$.
        In the evaluation of that subformula, $\neg x$ will then be required to be satisfied in the state where $x$ is bound.

        \par\smallskip\noindent
        \emph{For sufficiency of (\ref{total-sat-N-item1})--(\ref{total-sat-N-item4}),}
        assume that $\varphi$ meets one of the above conditions. It is not hard to see that the satisfaction of condition \eqref{total-sat-N-item1} implies that $\varphi$ is satisfied in a Kripke structure consisting of a single chain of states.
        Otherwise, if one of conditions \eqref{total-sat-N-item2}--\eqref{total-sat-N-item4} is satisfied, then $\varphi$ is satisfied in a Kripke structure consisting of two chains of states that share a common prefix of length $i_\xi-1$.
        
        \par\smallskip\noindent
        For the case that $\varphi$ contains unbound state variables or nominals,
        observe that $\varphi$ can be transformed into an equisatisfiable formula $\psi$ in polynomial time such that $\psi$ is free of nominals and unbound state variables:
        \[
          \psi := \dna r. \Diamond \dna x_1. \at_r
                          \cdots
                          \Diamond \dna x_k. \at_r
                          \Diamond \dna y_1. \at_r
                          \cdots
                          \Diamond \dna y_\ell. \at_r
                          \Diamond \varphi[i_1/y_1,\ldots,i_\ell/y_1],
        \]
        where $r,y_1,\ldots,y_\ell$ are fresh state variables,  
        $x_1,\ldots,x_k$ enumerate the unbound state variables occurring in $\varphi$,
        and $i_1,\ldots,i_\ell$ enumerate the nominals occurring in $\varphi$.
        If $\psi$ is satisfiable, then so is $\varphi$.
        On the other hand, if $\varphi$ is satisfiable in a total structure, then $\psi$ is satisfied in the structure
        obtained by adding a spypoint \cite{blse95}, i.e., a state $w$ that is a successor of every state of the original structure.
      \end{description}

      So far, an algorithm deciding satisfiability of $\varphi$ needs to check which of the three cases for $\lambda$
      holds. In the third case, after freeing $\varphi$ of nominals, unbound state variables, and multiple binding of the same
      state variable, $\modstring{\varphi}$ and $\modstring{\varphi'}$ need to be computed, and
      conditions \eqref{total-sat-N-item1}--\eqref{total-sat-N-item4} need to be checked. 
      In order to see that these tasks can be performed in logarithmic space,
      it needs to keep on mind that $\modstring{\varphi}$ and $\modstring{\varphi'}$ as well as the
      initial transformations need not be computed explicitly; 
      pointers to the current positions in $\varphi$, $\modstring{\varphi}$ and $\modstring{\varphi'}$ 
      suffice to determine the required information on-the-fly.

      \par\smallskip\noindent
      \textbf{Over ER frames, }%
      the above algorithm can be easily adapted.
      In case $\lambda \in \{\false, \true, p, \neg p,x,i\}$ we proceed as above. 
      Hence assume $\lambda \in \{\neg x,\neg i\}$. 
      
      First observe that the operator $\modstring{\cdot}$ has been defined for formulae without nominals and unbound state variables only.
      We extend its definition to the general case by letting $\modstring{\psi}$ be undefined whenever for some $\psi_j$ in the family of formulae needed to define $\modstring{\psi}$ there is no $\dna t.$, $t \in \NOM \cup \SVAR$, such that the suffix starting at $\at_t$ is free of $\at$-operators. 
       
      Now, if $\modstring{\varphi}$ and $\modstring{\varphi'}$ are both defined, then all states that $\varphi$ speaks about are situated in the same strongly connected component. It thus suffices to replace the conditions \eqref{total-sat-N-item1}--\eqref{total-sat-N-item4} imposed on $\modstring{\varphi}$ and $\modstring{\varphi'}$ with the following three.
      $\varphi$ is satisfiable over ER frames if and only if 
      \begin{enumerate}
        \item[(i')] $P_{i_\mu} = P_{j_\nu} = \Diamond$, or 
        \item[(ii')] $P_{i_\mu} = \Diamond$, $P_{j_\nu} = \Box$ and $j_\nu < i_\mu$.
        \item[(iii')] $P_{i_\nu} = \Diamond$, $P_{j_\mu} = \Box$ and $i_\nu < j_\mu$.
      \end{enumerate}
      If otherwise $\modstring{\varphi}$ or $\modstring{\varphi'}$ is undefined, we may assume that the state labeled by the nominal (resp.\ by the state variable) corresponding to $\lambda$ is bound in an equivalence class different from the equivalence class of the state that has to satisfy $\lambda$: assume, e.\,g., that $\modstring{\varphi'}$ is undefined due to $\at_t$, $t \in \NOM$. Then $K,g,w \models \varphi$ for the ER structure $K=(\{w,w'\},\{(w,w),(w',w')\}, \eta)$ with $[\eta,g](t)=w'$. The case that $\modstring{\varphi}$ is undefined is analogous.
      Hence $\varphi$ is satisfiable.
      
      \par\smallskip\noindent
      \textbf{Over arbitrary frames, }%
      it remains to extend the procedure for $\SaT[total]$ to also recognize formulae $\varphi$ being unsatisfiable over total but satisfiable over arbitrary or transitive frames (e.\,g., $\Box \false$). This is the case if and only if $\varphi$ contains a subformula $\Box\psi$ such that none of the states in which $\varphi$ is to be evaluated needs to have a successor (which might be enforced by preceding operators, consider, e.\,g., $\at_i \Diamond \at_i \Box \false$). More formally, for $1 < \ell \leq n$, let
      \[
        \modstring{P_1\cdots P_\ell\true}=P_{i_1}\cdots P_{i_\mu}\true.
      \]
      Then $\varphi$ is satisfiable over arbitrary frames if and only if 
      one of the conditions \eqref{total-sat-N-item1}--\eqref{total-sat-N-item4} is satisfied
      or for some $1 \leq \ell \leq n$, $P_{i_{\mu}} = \Box$ and for all $1 \leq k < \ell$,
        $k \in \{i_1,\ldots,i_\mu\}$ implies 
          that $\modstring{P_1\cdots P_k\true}=P_{j_1}\cdots P_{j_\nu}$ contains less than $\mu$ modalities 
          (i.\,e.\ $\mu > \nu$)
          or there exists a position $\xi$ such that $1 \leq \xi \leq \nu$ such that $P_{i_\xi}=P_{j_\xi}=\Diamond$.
      %

      \par\smallskip\noindent
      \textbf{Over transitive frames, }%
      we can combine the ideas above to an algorithm deciding the satisfiability for $B$-formulae.
  \end{proof}

\subsubsection{All Boolean functions}\label{subsubsec:all}

Finally, let us consider the clones between $\CloneD$, $\CloneS_1$ and $\CloneBF$, the clone of all Boolean functions.
For the classes of all frames, transitive frames, and equivalence relations, we can transfer results obtained for the set $\{\wedge,\vee,\neg\}$ of Boolean functions to these clones using a technical Lemma (see Appendix).
Additionally, we show that we can reduce the satisfiability problem over the class of all frames to the one over the class of total frames, establishing undecidability for all hybrid languages in this case.

    \begin{lemma}
      \label{lem:any_S1,D_Dia,dna}
      $\SaT[\Fclass{F}](O, \{\wedge,\vee,\neg\}) \leqPm \SaT[\Fclass{F}](O,B)$ for \allconsideredF,
      if $[B\cup \{\ONE\}]=\CloneBF$ and $O~\cap~\{\dna,\at\}~\not=~\emptyset$.
    \end{lemma}

    \begin{proof}
      Take $\varphi \in \HL(O, \{\wedge,\vee,\neg\})$. 
      Since $[B \cup \{\ONE\}] = [\{\wedge,\vee,\neg\}] = \CloneBF$,
      we can rewrite $\varphi$ as an $\HL(O, B \cup \{\ONE\})$-formula $\varphi'$,
      leaving modal and hybrid operators untouched.
      Due to \cite{lew79,sch07a}, this can be done in polynomial time.
      Now we can easily transform $\varphi'$ into an $\HL(O, B)$-formula $\varphi''$,
      replacing all occurrences of \ONE with $\dna x.x$ or $\at_x x$.
      Clearly, $\varphi$ and $\varphi''$ are equisatisfiable over $\Fclass{F}$.
    \end{proof}

    \begin{lemma}
      \label{lem:total_BF_Dia,dna}
        $\SaT(O, \{\wedge,\vee,\neg\}) \leqPm \SaT[total](O, \{\wedge,\vee,\neg\})$,
        for every set of operators $O\subseteq \{\Diamond,\Box,\dna,\at\}$.
    \end{lemma}

    \begin{proof}
      Let $B$ denote the set $\{\wedge,\vee,\neg\}$.
      We recursively define the reduction function
      $(\cdot)^r : \HL(O, B) \to \HL(O, B)$ as follows,
      \begin{xalignat*}{2}
        a^r                                      & = a,\quad a \in \ATOM                    & (\at_t\varphi)^r      & = \at_t(p \wedge \varphi^r)      \\
        \big(c(\varphi_1,\dots,\varphi_n)\big)^r & = c(\varphi_1^r, \dots, \varphi_n^r)     & (\dna x.\varphi)^r    & = \dna x.\varphi^r               \\
        (\Diamond\varphi)^r                       & = \Diamond(p \wedge \varphi^r)
      \end{xalignat*}
      where $p$ is a fresh atomic proposition and $c$ refers to an arbitrary $n$-ary Boolean operator.
      We show that, for any $\varphi \in \HL(O, B)$,
      $\varphi \in \SaT(O, B)$ if and only if $\varphi^r \in \SaT[total](O, B)$.

      For the ``$\Rightarrow$'' direction, assume $K,g,w \models \varphi$ for $K=(W,R,\eta)$.
      From $K$, we construct a total Kripke structure $K'=(W',R',\eta')$ with
      $W' = W \cup \{\tilde{w} \mid w \in W\}$,
      $R' = R \cup \{(w,\tilde{w}),(\tilde{w},\tilde{w}) \mid w \in W\}$,
      $\eta'(p) = W$, and
      $\eta'(x) = \eta(x)$ for all other atomic propositions and nominals $x$.
      It is straightforward to show inductively that $K',g,w \models \varphi^r$.

      For the ``$\Leftarrow$'' direction, assume that $K,g,w \models \varphi^r$ for a total Kripke structure $K=(W,R,\eta)$.
      From $K$, we construct a Kripke structure $K'=(W',R',\eta')$ with
      $W' = \{w\} \cup \eta(p) \cup \bigcup_{i \in \NOM} \eta(i) \cup \bigcup_{v \in \SVAR} \{g(v)\}$,
      $R' = R \cap W'$,
      $\eta'(x) = \eta(x) \cap W'$ for all other atomic propositions and nominals $x$.
      It is straightforward to show inductively that $K',g,w \models \varphi$.
    \end{proof}

    Lemmata \ref{lem:any_S1,D_Dia,dna} and \ref{lem:total_BF_Dia,dna},
    together with Theorem \ref{theo:known_for_and_or_neg},
    yield the following theorem.

    \begin{theorem}
      \label{theo:any_S1,D_Dia,dna}
      Let 
      $[B\cup\{1\}]=\CloneBF$.
      Then:
      \vspace{-2mm}
      \begin{Enum}
        \item
          $\SaT(O, B)$ and $\SaT[total](O, B)$ are \coRE-complete, for any $O\supseteq \{\Diamond,\dna\}$.
        \item
          $\SaT[trans](\{\Diamond,\dna,\at\}, B)$ is \coRE-complete.
        \item
          $\SaT[trans](\{\Diamond,\dna\}, B)$ is \NEXP-complete.
        \item
          $\SaT[ER](\{\Diamond,\dna\}, B)$ and
          $\SaT[ER](\{\Diamond,\dna,\at\}, B)$ are \NEXP-complete.
      \end{Enum}
    \end{theorem}